\documentclass[a4paper,12pt]{article}

\usepackage{amsmath}
\usepackage{a4wide}
\usepackage{amsfonts}
\usepackage{color}
\usepackage{amssymb,,amsmath}
\usepackage{graphicx}
\usepackage{amsfonts}
\newenvironment{proof}{\noindent\textbf{Proof:}}{\hspace{\stretch{1}}$\square$}

\newtheorem{theorem}{Theorem}[section]
\newtheorem{corollary}[theorem]{Corollary}
\newtheorem{definition}{Definition}
\newtheorem{proposition}[theorem]{Proposition}
\newtheorem{lemma}[theorem]{Lemma}

\newcommand{\rF}{\mathcal{F}}
\newcommand{\rD}{\mathcal{D}}
\newcommand{\rH}{\mathcal{H}}

\newcommand{\rK}{\mathcal{K}}
\newcommand{\rN}{\mathcal{N}}
\newcommand{\rM}{\mathcal{M}}
\newcommand{\rP}{\mathcal{P}}
\newcommand{\rE}{\mathcal{E}}
\newcommand{\rL}{\mathcal{L}}
\newcommand{\rZ}{\mathcal{Z}}

\newcommand{\rU}{\mathcal{U}}

\newcommand{\TF}{{T\Phi}}
\font\timesept=cmr7

\newcommand{\CC}{\mathbb{C}}

\newcommand{\EE}{\mathbb{E}}

\newcommand{\NN}{\mathbb{N}}

\newcommand{\PP}{\mathbb{P}}

\def\tr{\mathop{\rm Tr\,}\nolimits}
\def\O{\Omega}
\def\wt{\widetilde}
\def\indic{{\mathop{\rm 1\mkern-4mu l}}}
\def\qq{\qquad}
 
\def\ld{\ldots}
\def\wh{\widehat}
\def\wt{\widetilde}

\def\ol{\overline}

\def\indic{{\mathop{\rm 1\mkern-4mu l}}}

\def\tr{\mathop{\rm Tr\,}\nolimits}

\def\Tri{\mathop{\rm Tr}\nolimits}

\def\ecarte{\vphantom{\buildrel\bigtriangleup\over =}}

\def\ps#1#2{\langle #1\, ,\, #2\rangle}

\def\[{{\mathord{[\![}}}
\def\]{{\mathord{]\!]}}}

\def\norme#1{\left\| #1\right\|}
\def\normca#1{{\left\| #1\right\|}^2}
\def\ab#1{\left\vert #1\right\vert}

\def\x{\chi}

\def\O{\Omega}

\def\a{\alpha}

\def\s{\sigma}
\def\m{\mu}
\def\r{\rho}
\def\n{\nu}

\def\d{\delta}

\def\l{\lambda}

\def\o{\omega}
\def\wh{\widehat}

\begin{document}

\title{Classical Noises Emerging\\ from Quantum Environments\footnote{Work supported by ANR project ``StoQ" N${}^\circ$ ANR-14-CE25-0003}}

\author{S. Attal, J. Deschamps \&  C. Pellegrini}

\date{}

\maketitle

\begin{abstract}
In the framework of quantum open systems, that is, simple quantum systems coupled to quantum baths, our aim is to characterize those actions of the quantum environment which give rise to dynamics dictated by classical noises. First, we consider the discrete time scheme, through the model of repeated quantum interactions. We explore those unitary interactions which make the environment acting as if it were a classical noise, that is, dynamics which ought to random walks on the unitary group of the small system. We show that this characterization is intimately related to the notion of complex obtuse random variables. These particular random variables have an associated 3-tensor whose symmetries make it diagonalizable in some orthonormal basis. We show how this diagonalisation entirely describes the behavior of the random walk associated to the action of the environment. In particular this 3-tensor and its diagonalization characterize the behavior of the continuous-time limits; they characterize the directions along which the limit process is of diffusive or of Poisson nature.\end{abstract}

\vfill\eject
\tableofcontents

\section{Introduction}

Repeated quantum interaction models are physical models introduced and developed in \cite{A-P} which consist in describing the Hamiltonian dynamics of a quantum system undergoing a sequence of interactions with an  environment made of a chain of identical systems. These
models were developed for they furnish toy models for quantum
dissipative systems. They are at the same time Hamiltonian and
Markovian. They spontaneously give rise to quantum stochastic
differential equations in the continuous time limit.  It has been proved in \cite{B-P} and \cite{BdP} that they constitute a good toy model for a quantum heat bath in some situations and that they can also give an account of the diffusive behavior of an electron in an electric field, when coupled to a heat bath. When adding to each step of the dynamics a measurement of the piece of the environment which has just interacted, we recover all the discrete-time quantum trajectories for quantum systems (\cite{Pel}, \cite{Pel2}, \cite{Pel3}). Physically this model corresponds exactly to physical experiments such as the ones performed by S. Haroche et al.\thinspace on the quantum trajectories of a photon in a cavity (\cite{Har}, \cite{Har2}).

\smallskip
The discrete-time dynamics of these repeated interaction systems, as well as their continuous-time limit, give rise to time evolutions driven by quantum noises coming from the environment. These quantum noises emerging from the environment describe all the possible actions inside the environment (excitation, return to ground state, jumps in between two energy levels, ...). It is a remarkable fact that these quantum noises can also be combined together in order to give rise to classical noises. In discrete-time they give rise to any random walk, in continuous-time they give rise to many well-known stochastic processes among which are the Brownian motion, Poisson processes and all Levy processes. 

\smallskip
Our aim in this article is to characterize, among all possible evolutions driven by repeated quantum interactions, those which are driven by classical noises. This characterization is carried by the form of the basic unitary interaction in between the small system and one piece of the environment. In this article we concentrate on a large class of those \emph{classical unitary actions} of the environment. We show that they are intimately related to a particular algebraical structure: those 3-tensors on $\CC^N$ that can be diagonalized in some orthonormal basis, that is, the 3-tensor analogue of normal matrices (which are 2-tensors). 

We show that these diagonalizable 3-tensors are naturally associated to particular random variables on $\CC^N$: the \emph{complex obtuse random variables}, as were defined and developed in \cite{ADP}. They are a kind of basis for all the random variables in $\CC^N$. These particular random variables which appear as associated to the classical unitary actions of the environment are those which drive their classical action.

\smallskip
We then have a discussion on the continuous-time limit of this situation. This short discussion is based on the study of complex normal martingales, which are the continuous-time limits of complex obtuse random walks. 

\medskip
This article is constructed as follows. 
Section 2 is devoted to characterizing those unitary operators on bipartite systems which act classically on one of the two systems. There, we define the class of unitary operators on which we will concentrate. 


In Section 3 we recall the main properties of  complex obtuse random variables and their associated 3-tensors.

In Section 4 we establish our main result: we show that the diagonalization of 3-tensors gives an explicit representation of the repeated quantum interaction dynamics in terms of obtuse random walks. This representation gives an explicit description of the evolution of the small system as a random walk on its unitary group.

In Section 5 we present the mathematical setup associated to repeated quantum interactions and we apply the multiplication operator representation of classical unitary actions to this context; this gives rise to quantum dynamics driven by classical random walks. 

\smallskip
Section 6 is devoted to a discussion on the continuous-time limits of these classical unitary random walks, in terms of complex normal martingales.

\section{Classical Unitary Actions of Quantum Environments}

In this section we consider two quantum systems, with state space $\rH$ and $\rK$ respectively, which are interacting. We shall concentrate on those unitary operators on $\rH\otimes\rK$ which correspond to a classical action of the environment, that is, an action as a random unitary operator on $\rH$. 

In the whole article we suppose that $\rH$ and $\rK$ are finite dimensional. This is not a strictly necessary condition, as many of the points we discuss here could be extended to the infinite dimensional case.

\smallskip
Our approach is mainly motivated by the following statistical interpretation of density matrices. 

\subsection{Statistical Interpretation of Density Matrices}

In the quantum mechanics of open systems, states are represented by \emph{density matrices}, that is, positive, trace-class operators, with trace equal to 1. There are two ways to understand these density matrices, two \emph{interpretations}. 
First of all, the open quantum system point of view: the density matrices are what you see from a quantum system which is coupled to some environment, the whole coupled system being in a pure state. The following result is well-known.

\begin{theorem}\label{T:density_matrix}
Let $\r_\rH$ be an operator on a Hilbert space $\rH$. The following
assertions are equivalent. 

\smallskip
\noindent i) There exists a Hilbert space $\rK$ and a unit vector $\Psi$ in
   $\rH\otimes\rK$ such that $
\r_\rH=\Tri_\rK( \vert\Psi\rangle\langle\Psi\vert)\,.
$

\smallskip
\noindent ii)  The operator
$\rho_\rH$ is positive, trace-class and  $\tr(\rho_\rH)=1$\,. 
\end{theorem}

Density matrices can always be decomposed as 
\begin{equation}\label{E:densitymatrix}
\rho=\sum_{i=1}^n\l_i\,\vert e_i\rangle\langle e_i\vert
\end{equation}
for some  orthonormal basis $\{e_i\,;\, i=1,\ldots,n\}$ of $\rH$ and some positive
eigenvalues $\l_i$ satisfying 
$\sum_{i=1}^n\,\l_i=1$.
The density matrices represent the generalization of the notion of wave function which
is necessary to handle open quantum systems. Their decomposition under
the form \eqref{E:densitymatrix} can be understood as a mixture of wave
functions. 

\smallskip
Note the following extension of the representation of density matrices, which will be useful later. 

\begin{lemma}
If $\{\phi_i\,;\,i=1,\ldots,m\}$ is any family of norm 1 vectors in $\rH$ and if $\{\l_i\,;\,i=1,\ldots,m\}$ are positive reals such that $\sum_{i=1}^K \l_i=1$, then the operator
$
\r=\sum_{i=1}^m \l_i\,\vert\phi_i\rangle\langle\phi_i\vert
$
is a density matrix on $\rH$.
\end{lemma}

The difference with the representation \eqref{E:densitymatrix} lies is the fact that the $\phi_i$'s need not be pairwise orthogonal. In particular there can be more of them than the dimension of the state space, or they can be linearly dependent.

\bigskip
The second point of view on density matrices is that they can be understood as \emph{random pure states}. Let us describe this interpretation and explain why it represents the same states as density matrices.
First of all recall the main axioms of Quantum Mechanics concerning states. 
The axioms of Quantum Mechanics say that measuring an observable $X$ (that is, a self-adjoint operator), with spectral measure $\xi_X$, gives a value for the measurement which lies in the set $A$ with probability 
$\tr(\r\,\xi_X(A))\,.$

After having measured that the observable $X$ lies inside the set $A$, the state $\r$ is transformed into the state
$$
\frac{\xi_X(A)\,\r\,\xi_X(A)}{\tr(\r\,\xi_X(A))}\,.
$$
This is the so-called ``\emph{reduction of the wave packet}''.

The last basic axiom concerning states in Quantum Mechanics is the one which describes their time evolution under the Schr\"odinger dynamics: if $H$ is the Hamiltonian of the system and $\r_0$ is the state at time 0, then the state at time $t$ is $\r_t=U_t\,\r_0\,U_t^*$, where $U_t=\exp(-i tH)$. 

\smallskip
We shall now compare these elements for density matrices with the ones we obtain with random states. 
Consider a random pure state $\vert \phi\rangle$ which is equal to $\vert \phi_i\rangle$ with probability $\l_i$, $i=1,\ldots, n$. Note that the $\phi_i$'s are all norm 1 but not necessarily orthogonal with each other.  

\begin{proposition}\label{P:stat_dens}
Let $\vert \phi\rangle$ be a random pure state which is equal to $\vert \phi_i\rangle$ with probability $\l_i$, $i=1,\ldots, m$. Let $\r$ be the density matrix 
$
\r=\sum_{i=1}^m \l_i\, \vert\phi_i\rangle\langle\phi_i\vert\,.
$
Then measuring any observable $X$ with the random pure state $\vert\phi\rangle$ gives a measure lying in a set $A$ with probability
$\tr(\r\,\xi_X(A))\,.$

Furthermore, after having measured that $X$ lies in $A$, the random state $\vert\phi\rangle$ is transformed into another random pure state $\vert\psi\rangle$ whose associated density matrix is
$$
\frac{\xi_X(A)\,\r\,\xi_X(A)}{\tr(\r\,\xi_X(A))}\,.
$$

Finally, under the unitary evolution associated to $H$, the random state $\vert \phi\rangle$, at time 0,  becomes at time $t$ a random pure state $\vert\phi_t\rangle$ whose associated density matrix is $U_t\, \r\,U_t^*$. 
\end{proposition}

\begin{proof}
Let us detail what happens if we measure the observable $X$ with the random state $\vert\phi\rangle$. Adopting obvious probabilistic notations, we know that
$\PP(\mbox{choose }\phi_i)=\l_i$ {and} 
$$
\PP(\mbox{measure in }A\,\vert\,\mbox{choose }\phi_i)=\normca{\xi_X(A)\, \vert \phi_i\rangle}\,=\ps{\phi_i}{\xi_X(A)\, \phi_i}\,.
$$
In particular, 
\begin{align*}
\PP(\mbox{measure in }A)&=\sum_{i=1}^m \PP(\mbox{measure in }A\,\vert\,\mbox{choose }\phi_i)\,
\PP(\mbox{choose }\phi_i)\\
&=\sum_{i=1}^m\ps{\phi_i}{\xi_X(A)\, \phi_i}\,\l_i=\tr(\r\,\xi_X(A))\,.
\end{align*} 
Hence measuring any observable $X$ of $\rH$ with the density matrix $\r$ or with the random state $\vert\phi\rangle$ gives the same probabilities. 

\smallskip
Let us compute the reduction of the wave packet in this case too. If we have measured $X$ with the random state $\vert\phi\rangle$ we have measured $X$ with one of the state $\vert \phi_i\rangle$. After having measured the observable $X$ with a value in $A$, the state of the system is one of the pure states
$$
\frac{\xi_X(A)\vert \phi_i\rangle}{\norme{\xi_X(A)\vert \phi_i\rangle}}\,,
$$
with probability 
\begin{align*}
\PP(\mbox{choose }\phi_i\,\vert\,\mbox{measure in }A)&=
\frac{\PP(\mbox{measure in }A,\ \mbox{choose }\phi_i)}{\PP(\mbox{measure in }A)}\\
&=\frac{\l_i\,\ps{\phi_i}{\xi_X(A)\, \phi_i}}{\tr(\r\,\xi_X(A))}=\frac{\l_i\,\normca{\xi_X(A)\, \phi_i}}{\tr(\r\,\xi_X(A))}\,.
\end{align*}
Hence, after the measurement we end up with a random pure state again, which corresponds to the density matrix
$$
\sum_{i=1}^m \frac{\l_i\normca{\xi_X(A)\, \phi_i}}{\tr(\r\,\xi_X(A))}\, \frac{\xi_X(A)\vert \phi_i\rangle\langle \phi_i\vert \xi_X(A)}{\normca{\xi_X(A)\vert \phi_i\rangle}}=\sum_i \l_i\,\frac{\xi_X(A)\vert \phi_i\rangle\langle \phi_i\vert \xi_X(A)}{\tr(\r\,\xi_X(A))}\,,
$$ 
that is, the density matrix
$$
\frac{\xi_X(A)\,\r\,\xi_X(A)}{\tr(\r\,\xi_X(A))}\,.
$$

The statement about the unitary evolution is obvious. 
\end{proof}

\bigskip
In that sense, from the physicist point of view, these two states, when measuring observables (which is what states are meant for) give exactly the same results, the same values, the same probabilities, the same reduction of the wave packet; they also have the same time evolution. Hence they describe the same ``\emph{state}" of the system.
This equivalence is only physical, in the sense that whatever one wants to do physically with these states, they give the same results. Mathematically they are clearly different objects: one is a positive trace 1 operator, the other one is a random pure state.

\smallskip
Note that the statistical interpretation of density matrices can be also understood in another sense, namely in the sense of \emph{indirect measurement}. Indeed, if $\r$ is a density matrix of the form
$$
\r=\sum_{i=1}^m \l_i\, \vert\phi_i\rangle\langle\phi_i\vert\,,
$$
then one can consider the Hilbert space $\rK=\CC^m$ and the following pure state on $\rH\otimes\rK$:
$$
\psi=\sum_{i=1}^m \sqrt{\l_i} \, \phi_i\otimes f_i\,,
$$
where $\{f_i\,;\, i=1,\ldots,m\}$ is an orthonormal basis of $\rK$. Now, measuring on $\rK$ along this basis, gives the state $\phi_i$ on $\rH$, with probability $\l_i$.

\subsection{Classical Unitary Actions}

We now focus on unitary interactions between two quantum systems. We have the two quantum systems $\rH$ and $\rK$ interacting with each other, their dynamics is driven by a total Hamiltonian
$$
H=H_S\otimes I+I\otimes K+H_{\rm int}\,.
$$
Both systems evolve this way during a time interval of length $h$, so that the evolution of the state of the system is driven by the unitary operator
$$
U=e^{-ihH}\,,
$$
which is a rather general unitary operator $U$ on $\rH\otimes \rK$. 

\smallskip
We are interested in the resulting action of $U$ on $\rH$. This action is described as follows. Consider any given state $\o$ on $\rK$, we compute the mapping
$$
\rL(\r)=\tr_{\rK}\left(U(\rho\otimes\o)U^*\right)
$$
for all density matrices $\rho$ of $\rH_S$. This mapping exactly expresses what we recover from the system $\rH$ of the action of the environment $K$ via the unitary $U$.
This mapping $\rL$ is well-known to be a \emph{quantum channel} on $\rH$ and it admits a so-called \emph{Krauss representation}, that is, there exist bounded operators $L_i$, $i\in I$, on $\rH$ such that
$
\sum_{i\in I} L_i^*L_i=I
$ 
and
$$
\rL(\rho)=\sum_{i\in I} L_i\, \rho\, L_i^*
$$
for all density matrices $\rho$.

The Krauss representation of $\rL$ as above is not unique, there is a freedom in the way of choosing the coefficients $L_i$. This non-uniqueness is completely described by the following well-known theorem. 

\begin{theorem}[GHJW Theorem]\label{T:GHJW}
The two Krauss representations 
$
\rL(\r)=\sum_{i=1}^l L_i\, \r\, L_i^*
$
and 
$
\rM(\r)=\sum_{i=1}^k M_i\, \r\,M_i^*\,,
$
with $l\leq k$ say, represent the same quantum channel if and only if there exists a unitary matrix $(u^i_j)_{i,j=1,\ldots, k}$ such that
$
M_i=\sum_{j=1}^k u^i_j\, L_j\,,
$
where the list of the $L_j$'s has been completed by zeros if $l<k$. 
\end{theorem}

Now we characterize those unitary actions that we call \emph{classical}. We begin with an important remark.

\begin{proposition}
Let $\rL$ be a quantum channel on $\rH$ which admits a Krauss decomposition
$$
\rL(\r)=\sum_{i=1}^m L_i\, \r\, L_i^*
$$
where the $L_i$'s are all scalar multiples of unitary operators: $L_i=\l_i\, U_i$ on $\rH$. In particular we have
$
\sum_{i=1}^m \ab{\l_i}^2=1\,.
$
Then the action of $\rL$ on the states of $\rH$ is exactly the same as if we apply a random unitary transform to $\rH$: choosing one of the unitary actions $U_i$ with respective probability $\ab{\l_i}^2$\,.
\end{proposition} 

\begin{proof}
According to the rules of Quantum Mechanics, if the unitary transform $U_i$ is applied to $\rH$ then the state $\r$ of $\rH$ is transformed into $U_i\, \r\, U_i^*$. If this occurs randomly with probability $p_i$, respectively, then the new state of $\rH$ is one of the states $U_i\, \r\, U_i^*$ with probability $p_i$. That is, according to the statistical interpretation of density matrices developed above, we end up with the state 
$$
\sum_{i=1}^m p_i\, U_i\, \r\, U_i^*=\sum_{i=1}^m (\sqrt{p_i}\,U_i)\, \r\, (\sqrt{p_i}\, U_i)^*=\rL(\r)\,.
$$
This gives the result.
\end{proof}

\smallskip
Note that we have the same interpretation of the result above in terms of \emph{indirect measurements}. Indeed, consider the space $\rK=\CC^m$, with an orthonormal basis $\{f_i\,;\, i=1,\ldots, m\}$. On the space $\rH\otimes\rK$ there exists a unitary operator $V$ whose first (block) column in the basis $(f_i)$ is given by the coefficients $L_i=\sqrt{\l_i}\, U_i$. Consider any density matrix $\r$ on $\rH$ and evolve the state $\r\otimes\vert f_1\rangle\langle f_1\vert$ with $V$, this gives rise to the state
$$
V\, (\r\otimes\vert f_1\rangle\langle f_1\vert)\, V^*\,.
$$
Now, measuring $\rK$ along the basis $(f_i)$ gives the state $U_i\,\r\,U_i^*$ with probability $\l_i$. 

\bigskip
This motivates the following definition of a classical unitary action from $\rK$ onto $\rH$.
\begin{definition}
A unitary operator $U$ on $\rH\otimes\rK$ gives rise to a \emph{classical action of $\rK$ onto $\rH$} if for every state $\o$ on $\rK$, the quantum channel
$$
\rL(\r)=\tr_{\rK}\left(U(\r\otimes\o)U^*\right)
$$
on $\rH$ admits a Krauss decomposition made of multiples of unitary operators only.
\end{definition}

For short we may say that $U$ is a \emph{classical unitary}, but one has to keep in mind that it is classical for its action on $\rH$ only, that is, the definition above is not symmetric in the roles of $\rH$ and $\rK$.

\smallskip
Due to the freedom left by the GHJW Theorem on the Krauss representation of a given quantum channel, the above definition of a classical unitary action seems very difficult to characterize. We have not been able to completely characterize it. There is a very strong condition in the above definition: it is the fact that this property is asked for every state $\o$ on $\rK$. But the strength of this condition is compensated by the fuzziness of the sentence ``there exists a Krauss decomposition made of ...", which is well-known to be very hard to characterize in general.

Furthermore, it is not clear if one should ask the $U_i$'s to be independent of $\o$ or not, if one should ask the $\l_i$'s to be linear in $\o$ or not. In \cite{DNP}, results in this direction have been addressed and partial results have been obtained.

\smallskip
In this article we concentrate on a very large class of unitary actions which are classical in that sense. We are strongly convinced and we conjecture that they are the only possible classical unitaries, but we have not been able to prove this fact.

\begin{proposition}\label{P:stat_unit}
If  $U$ is of the form
$$
U=\sum_{i=1}^k U_i\otimes \vert e_i\rangle\langle f_i\vert
$$
for some unitary operators $U_i$ on $\rH$, some orthonormal bases $(e_i)$ and $(f_i)$ of $\rK$, then $U$ is a classical unitary action from $\rK$ onto $\rH$.
\end{proposition}

\begin{proof}
First of all, an operator of the form above is a unitary operator, for
\begin{align*}
U^*U&=\sum_{i,j} U_i^*U_j\otimes\vert f_i\rangle\langle e_i\vert\,\vert e_j\rangle\langle f_j\vert=\sum_i U_i^*U_i\otimes \vert f_i\rangle\langle f_i\vert\\
&=\sum_i I\otimes \vert f_i\rangle\langle f_i\vert=I\otimes I\,,
\end{align*}
and the same holds in the same way for $UU^*$.

Let $\r$ be some initial state on $\rH$ and $\o$ be the state of $\rK$ before interaction. After interaction the state becomes
\begin{align*}
U(\r\otimes\o)U^*&=\sum_{i,j} U_i\,\r\, U_j^*\otimes \vert e_i\rangle\langle f_i\vert \,\o\, \vert f_j\rangle\langle e_j\vert \\
&=\sum_{i,j}\langle f_i\vert \,\o\, \vert f_j\rangle\, U_i\,\r\,U_j^*\otimes \vert e_i\rangle\langle e_j\vert \,.
\end{align*}
The resulting state on $\rH_S$ is 
$$
\rL(\r)=\tr_{\rK}\left(U(\r\otimes\o)U^*\right)=\sum_{i}\langle f_i\vert \,\o\, \vert f_i\rangle\, U_i\, \r\,U_i^*\,.
$$
This is exactly the required form for $U$ to be a classical unitary action.
\end{proof}

\smallskip
Note that, from the proof above, we can immediately deduce the following important simplification.

\begin{corollary}
If  $U$ is of the form
$$
U=\sum_{i=1}^k U_i\otimes \vert e_i\rangle\langle f_i\vert
$$
for some unitary operators $U_i$ on $\rH$, some orthonormal bases $(e_i)$ and $(f_i)$ of $\rK$, then $U$ has the same classical unitary action on $\rH$ as the operator
\begin{equation}\label{E:finalU}
V=\sum_{i=1}^k U_i\otimes \vert f_i\rangle\langle f_i\vert\,.
\end{equation}
\end{corollary}

As a consequence we shall now concentrate on those unitary actions of the form (\ref{E:finalU}) only. 

\smallskip
The aim of this article is to completely characterize such ``random unitary operators'' (or \emph{classical unitary interactions}) on $\rH\otimes\rK$ and their continuous-time limits, to express them in terms of multiplication operators by a natural basis of random processes, namely the \emph{complex obtuse random variables} in discrete time and \emph{the complex normal martingales} in continuous time.

\smallskip
Once again, note that $U$ is not explicitly a random unitary operator, that is, a random variable with values in the unitary group $\rU(\rH)$. It is a deterministic unitary operator on $\rH\otimes\rK$ whose action on the states of $\rH$ oughts to the same quantum physics as the action of a random unitary operator on $\rH$. In the same way as we have discussed previously: density matrices have two different interpretations, describing the same physics.

\section{Complex Obtuse Random Variables}

We shall show that the structure of those classical unitary interactions is intimately related to a particular family of random variables, the \emph{complex obtuse random variables}. Real obtuse random variables were introduced in \cite{A-E}, for they appear naturally in the theory of discrete-time normal martingales, predictable and chaotic representation property, for vector-valued martingales. The extension to the complex case was studied in \cite{ADP}, for it is more compatible with the physical situation. Note that the extension from the real to the complex case is far from obvious. 

This section is devoted to the presentation the main results of \cite{ADP} on complex obtuse random variables and their associated $3$-tensors which shall be necessary for the description of classical unitary interactions.

\subsection{Obtuse Systems and Obtuse Random Variables}

\begin{definition}
An \emph{obtuse system} in $\CC^N$ is a family of $N+1$ vectors $v_1,\ldots, v_{N+1}$ such that
$$
\ps{v_i}{v_j}=-1
$$
for all $i\not =j$. 
This family can be embedded into an orthogonal family of $N+1$ vectors in $\CC^{N+1}$ defined by
$$
\wh v_i=\left(
\begin{matrix}
 1 \\
 v_i   
\end{matrix}
\right)\in\CC^{N+1}\,.
$$
To obtuse systems, one can associate a natural probability distribution
$$
p_i=\frac{1}{\normca{\wh v_i}}=\frac 1{1+\normca{v_i}}\,,
$$
giving the following properties
\begin{equation}\label{proba}
\sum_{i=1}^{N+1} p_i=1\,,\quad  \sum_{i=1}^{N+1} p_i\,v_i=0\,,\quad\text{and }\,,\quad \sum_{i=1}^{N+1} p_i\,\vert v_i\rangle\langle v_i\vert=I_{\CC^N}\,.
\end{equation}
\end{definition}

\begin{proposition}\label{unitary}
Let $\{v_1,\ldots, v_{N+1}\}$ be an obtuse system of $\CC^N$ having $\{p_1,\ldots, p_{N+1}\}$ as associated probabilities. Then the following assertions are equivalent.

\smallskip\noindent
i) The family $\{w_1,\ldots, w_{N+1}\}$ is an obtuse system on $\CC^N$ with same respective probabilities $\{p_1,\ldots, p_{N+1}\}$.

\smallskip\noindent 
ii) There exists a unitary operator $U$ on $\CC^N$ such that $w_i=Uv_i$, for all $i=1,\ldots, N+1$.  
\end{proposition}

\begin{definition}
An \emph{obtuse random variable} $X$ in $\CC^{N}$ is a random variable taking its values $v_1,\ldots,v_{N+1}$ in an obtuse system associated to its natural probability distribution $p_1,\ldots, p_{N+1}$. Because of \eqref{proba}, the random variable $X$ is \emph{centered}, that is, we have $\EE[X^i]=0$ for all $i$, and \emph{normalized}, that is, we have
$$
\mbox{\rm cov}(X^i,X^j)=\EE[\ol{X^i}\, X^j]-\EE[\ol{X^i}]\,\EE[X^j]=\delta_{i,j}\,,
$$
where $X^1, \ldots, X^N$ denote the coordinates of $X$ in $\CC^N$. 

The random variable $X$ can be defined on its canonical space $(\O, \rF, \PP)$, that is, $\O=\left\{1, \dots, N+1\right\}$, $\rF$ is the full $\s$-algebra of $\O$, the probability measure $\PP$ is given by $\PP \left(\left\{i\right\}\right)= p_i$ and $X$ is given by $X(i)=v_i$. The coordinates of $v_i$ are also denoted by $v_i^k$, for $k=1,\ldots, N$, so that $X^k(i)=v^k_i$. 
\end{definition}

We shall also consider the deterministic variable $X^0$ on $(\O, \rF, \PP)$, which is always equal to $1$. 

%
%
%
%
%

\subsection{Associated 3-Tensors}

We present now how obtuse random variables are naturally related to some 3-tensors with particular symmetries. These 3-tensors are essential to characterize the classical unitary evolutions, as well as their continuous time limit in the scheme of repeated quantum interactions.

\begin{definition}
A \textit{3-tensor} $S$ on $\CC^{N+1}$  is a linear map from $\CC^{N+1}$ to $\CC^{N+1}\otimes \CC^{N+1}$, that is, a collection of coefficients $(S^{ij}_k)_{i,j,k=0}^{N}$ defining the following action on $\CC^{N+1}$:
$$
(S(x))^{ij} = \sum_{k=0}^{N} S^{ij}_k x^k\,.
$$
\end{definition}

We shall see below that obtuse random variables on $\CC^N$ have a naturally associated 3-tensor on $\CC^{N+1}$. This associated 3-tensor is related to the notion of complex doubly-symmetric 3-tensor.

\begin{definition} 
A 3-tensor $(S^{ij}_k)_{i,j,k=0}^N$ on $\CC^{N+1}$ is called a complex doubly symmetric $3$-tensor if it satisfies 
\begin{equation}\label{E:sym1}
S^{ij}_k \ \ \mbox{is symmetric in}\ \ (i,j)\,,
\end{equation}
\begin{equation}\label{E:sym2}
\sum_{m=0}^N {S^{im}_j}\,S^{kl}_m\ \ \mbox{is symmetric in}\ \ (i,k)\,,
\end{equation}
\begin{equation}\label{E:sym3}
\sum_{m=0}^N S^{im}_j\,\ol{S^{lm}_k}\ \ \mbox{is symmetric in}\ \ (i,k)\,.
\end{equation}
\end{definition}

Complex doubly-symmetric 3-tensors happen to be exactly those 3-tensors which are diagonalizable in some orthonormal basis. That is, the exact generalization to 3-tensors of being a normal matrix for 2-tensors. This non-trivial result is proved in \cite{ADP}. 
One can also express this diagonalization in terms of obtuse random variables, this is the theorem we shall retain here.

\begin{theorem}\label{P:STdiscret}
Let $X$ be an obtuse random variable in $\CC^N$. Then there exists a unique complex doubly symmetric 3-tensor $S$ on $\CC^{N+1}$ such that
\begin{equation}\label{E:XX=TX}
X^i\, X^j=\sum_{k=0}^N S^{ij}_k\, X^k\,,
\end{equation}
for all $i,j=0,\ldots, N$. 
This complex doubly symmetric $3$-tensor $S$ is given by
\begin{equation}\label{E:Tijk}
S^{ij}_k=\EE[{X^i}\, X^j\, \ol{X^k}]\,,
\end{equation}
for all $i,j,k=0,\ldots N$ and satisfies also
\begin{equation}\label{E:sym0}
S^{i0}_k=\d_{ik}\,.
\end{equation}

Conversely, to each complex doubly symmetric $3$-tensor on $\CC^N$ satisfying  \eqref{E:sym0}, one can attach a unique (in law) obtuse random variable $X$ on $\CC^N$ such that (\ref{E:XX=TX}) and (\ref{E:Tijk}) hold.
\end{theorem}

Another useful relation satisfied by the complex doubly symmetric $3$-tensor attached to obtuse random variable $X$  is the following:
\begin{equation}\label{E:XbX}
\ol{X^i}\,X^j=\sum_{k=0}^N \ol{S^{ik}_j}\, X^k\,.
\end{equation}
for all $i,j=0,\ldots, N$ 

\smallskip
Now we have presented the relation between obtuse random variables and doubly-symmetric $3$-tensors, we describe the main motivation of the introduction of $3$-tensors which is the representation of multiplication operators.

\subsection{Representation of Multiplication Operators}\label{SS:mult}

\begin{definition}
Let $X$ be an obtuse random variable in $\CC^N$, with associated 3-tensor $S$ and let $(\O,\rF,\PP_X)$ be the canonical space of $X$. The space $L^2(\O,\rF,\PP_X)$ is a $N+1$-dimensional Hilbert space and the family $\{{X^0}, {X^1},\ldots, {X^N}\}$ is an orthonormal basis of that space. Hence for every obtuse random variable $X$, we have a natural unitary operator
$$
\begin{matrix} W_X&:&L^2(\O,\rF,\PP_X)&\longrightarrow&\CC^{N+1}\\
&&{X^i}&\longmapsto&e_i\,,
\end{matrix}
$$
where $\{e_0,\ldots, e_N\}$ is the canonical orthonormal basis of $\CC^{N+1}$. The operator $W_X$ is called the \emph{canonical isomorphism} associated to $X$. 
\end{definition}

The interesting point with these isomorphisms $W_X$ is that they canonically transport all the obtuse random variables of $\CC^N$ onto a common canonical space. The point is that the probabilistic informations concerning the random variable $X$ are not correctly transferred via this isomorphism: all the informations about the law, the independencies, ... are lost when identifying $X^i$ to $e_i$. The only way to recover the probabilistic informations about the $X^i$'s on $\CC^{N+1}$ is to consider the \emph{multiplication operator} by $X^i$, defined as follows.

\begin{definition}
On the space $L^2(\O,\rF,\PP_X)$, for each $i=0,\ldots, N$, we consider the multiplication operator
$$
\begin{matrix} \rM_{X^i}&:&L^2(\O,\rF,\PP_X)&\longrightarrow&L^2(\O,\rF,\PP_X)\\
&&{Y}&\longmapsto&X^i\,Y\,.
\end{matrix}
$$
\end{definition}

\begin{definition}
On the space $\CC^{N+1}$, with canonical basis $\{e_0,\ldots, e_N\}$ we consider the basic matrices $a^i_j$, for $i,j=0,\ldots, N$ defined by
$$
a^i_j\, e_k=\delta_{i,k}\, e_j\,.
$$
\end{definition}

We shall see now that, when carried out on the same canonical space by $W_X$, the obtuse random variables of $\CC^N$ admit a simple and compact matrix representation in terms of their 3-tensor.

\begin{theorem}\label{T:multop}
Let $X$ be an obtuse random variable on $\CC^N$, with associated 3-tensor $S$ and canonical isomorphism $W_X$. Then we have, for all $i=0,\ldots, N$
\begin{equation}\label{UXU}
W_X\, \rM_{X^i}\, W_X^*=\sum_{j,k=0}^N S^{ij}_k\, a^j_k\,.
\end{equation}
The operator of multiplication by $\ol{X^i}$ is given by
\begin{equation}\label{olXX}
W_X\,\rM_{\ol{X^i}}\, W_X^*=\sum_{j,k=0}^N \ol{S^{ik}_j}\, a^j_k\,.
\end{equation}
\end{theorem}

\section{Representation of Classical Unitary Actions}

Coming back to classical unitary actions of $\rK$ on $\rH$, it happens that obtuse random variables appear naturally in the representation of the unitary operator. This is the main result of this article. 

\subsection{Multiplication Operator Representation}

For the moment we assume that the reference state of $\rK$ in which the probabilities of $U$ are computed is the pure state $e_0$, first vector of the canonical orthonormal basis of $\rK$. Moreover, we assume that $\ps{\phi_i}{e_0}\not=0$ for all $i$. 

We shall see afterwards how all the other cases for the reference state $\o$ on $\rK$ can be boiled down to this particular case. 

\begin{theorem}\label{representation}
On the space $\rK=\CC^{N+1}$, let $U=\sum_{i=1}^{N+1} U_i\otimes \vert\phi_i\rangle\langle\phi_i\vert$ be a unitary operator on $\rH\otimes \rK$ acting classically on $\rH$. Assume that $\ps{\phi_i}{e_0}\not=0$ for all $i$. Then there exists an obtuse random variable $X$ on $\CC^N$, with associated 3-tensor $S$ and canonical isomorphism $W_X$, such that
\begin{equation}\label{E:reprU}
U=A\otimes I+\sum_{j=1}^N B_j\otimes W_X\,\rM_{X^j}\, W_X^*
\end{equation}
for some operators $A, B_1, \ldots, B_N$ on $\rH$. 
 
More precisely, the obtuse random variable $X$ is the one associated to the  probability distribution $\{\ab{\ps{e_0}{\phi_i}}^2\,,\ i=1,\ldots, N+1\}$ and to the values
\begin{equation}\label{vv}
v_i^j=\frac{\ps{\phi_i}{e_j}}{\ps{\phi_i}{e_0}}\,,
\end{equation}
$i=1,\ldots, N+1$, $j=1,\ldots, N$.  The operators $A$ and $B_j$ are given by
\begin{equation}\label{E:reprUA}
A=\sum_{i=1}^{N+1} p_i\,U_i
\end{equation}
and 
\begin{equation}\label{E:reprUB}
B_j=\sum_{i=1}^{N+1} p_i\,\ol{v^j_i}\,U_i\,.
\end{equation}
\end{theorem}

One can also notice the following equalities:
\begin{eqnarray}
B_j&=&\sum_{i=1}^{N+1} p_i\,\ol{v^j_i}\,U_i\,\\
&=&\sum_{i=1}^{N+1}\ps{\phi_i}{e_j}\overline{\ps{\phi_i}{e_0}}\,U_i\\
&=&\sum_{i=1}^{N+1}U_i\,\langle e_0\vert\vert \phi_i\rangle\langle \phi_i\vert\vert e_j\rangle\\
&=&Tr_\mathcal K[U\,I\otimes \vert e_j\rangle\langle e_0\vert]\\
A&=&Tr_\mathcal K[U\,I\otimes \vert e_0\rangle\langle e_0\vert]
\end{eqnarray}

\smallskip\noindent
{\bf Remark:} Before proving the theorem above, it is important to make the following remark. The obtuse system, the operators $B_j$, described above, seem to depend highly on the choice of the orthonormal basis $\{e_0, e_1,\ldots,e_N\}$ of $\CC^{N+1}$ which has been chosen at the begining. Let us show that this is not the case. The construction above depends clearly on the choice of $e_0$, as it corresponds to the choice of the state of the environment, that is, it dictates the probability distribution one should find, but it does not depend on the way $e_0$ is completed into an orthonormal basis of $\CC^{N+1}$. 

In order to make that clear, let us see what happens if we change the orthonormal basis $\{e_0, e_1,\ldots,e_N\}$ into another one $\{e_0, f_1,\ldots,f_N\}$. In particular there exists a unitary operator $U$ of $\CC^N$ such that $f_j=U\,e_j$ for all $j=1,\ld, N$; or else $f_j=\sum_{k=1}^N U_{kj}\, e_k$. Let us denote by $\wt X^i$ the coordinates of $X$ in this new basis $\{f_1,\ldots, f_N\}$ of $\CC^N$, that is, $\wt X^i=\sum_{k=1}^N U_{kj}\, X^k$. 

The canonical operator $W\,:\, \wt X^i\mapsto f_i$ remains the same as $W_X$, as can be checked easily. The probabilities $p_i$, obtained above, remain the same for they depend only on $e_0$. The obtuse system $\{v_1, \ldots, v_{N+1}$ is changed as follows:
$$
w^j_i=\frac{\ps{\phi_i}{f_j}}{\ps{\phi_i}{f_0}}=\sum_{k=1}^N U_{kj}\,\frac{\ps{\phi_i}{e_k}}{\ps{\phi_i}{e_0}}=
\sum_{k=1}^N U_{kj}\,v^k_i\,.
$$
In other words, all the $w_i$'s are image of the corresponding $v_i$'s under a unitary map of $\CC^N$. As was explained in Proposition \ref{unitary}, this kind of transformation links all obtuse systems with same associated probability measure.

The operators $A_j$ are unchanged, but the operators $B_j$ are also changed into
$$
\wt B_j=\sum_{k=1}^N \ol{U_{kj}}\, B_k\,.
$$
The representation \eqref{E:reprU} then becomes
\begin{align*}
U&=A\otimes I+\sum_{j=1}^N \wt B_j\otimes W_X\,\rM_{\wt X^j}\, W_X^*\\
&=A\otimes I+\sum_{j=1}^N \sum_{k,l=1}^N \ol{U_{kj}} \, U_{lj} B_k\otimes W_X\,\rM_{ X^l}\, W_X^*\\
&=A\otimes I+ \sum_{k,l=1}^N \left(\sum_{j=1}^N\ol{U_{kj}} \, U_{lj}\right) B_k\otimes W_X\,\rM_{ X^l}\, W_X^*\\
&=A\otimes I+ \sum_{k=1}^N  B_k\otimes W_X\,\rM_{ X^k}\, W_X^*\,,
\end{align*}
for $U$ is unitary. 

We recover the same representation as \eqref{E:reprU}. This ends the discussion as follows: only the choice of $e_0$ is important, it determines the probability measure of the associated obtuse random variable $X$. Changing the rest of the orthonormal basis, only changes the choice of the obtuse system, in the class of those sharing the same probability measure. The representation \eqref{E:reprU} is independant of this change of basis. 

\medskip
We now prove the theorem.
 
\begin{proof}
In the canonical basis of $\rK$, the vector $\phi_i$ is of the form 
$$
\phi_i=
\left(\begin{matrix}
 \ps{e_0}{\phi_i}   \\
  \vert w_i \rangle
\end{matrix}\right)
=\ps{e_0}{\phi_i}\left(\begin{matrix}
 1   \\
  \vert \ol{v_i}\rangle 
\end{matrix}\right)\,,
$$
where $\vert v_i\rangle$ is the vector of $\CC^N$ described in \eqref{vv}.
Putting $p_i=\ab{\ps{e_0}{\phi_i}}^2$, this gives
$$
\vert \phi_i\rangle\langle \phi_i\vert=p_i\left(\begin{matrix}1&\langle \ol{v_i}\vert\\\\\vert \ol{v_i}\rangle&\vert \ol{v_i}\rangle\langle\ol{v_i}\vert\end{matrix}\right).
$$
We can write it also as
$$
\vert \phi_i\rangle\langle \phi_i\vert=p_i\,I+p_i\,\left(\begin{matrix}0&\langle \ol{v_i}\vert\\\\\vert \ol{v_i}\rangle&\vert \ol{v_i}\rangle\langle \ol{v_i}\vert-I\end{matrix}\right).
$$
We shall concentrate on the action of the last matrix. On a generic vector it acts as follows
$$
\left(\begin{matrix}0&\langle \ol{v_i}\vert\\\\\vert \ol{v_i}\rangle&\vert \ol{v_i}\rangle\langle \ol{v_i}\vert-I\end{matrix}\right)\left(\begin{matrix} x\\\\\vert \o\rangle\end{matrix}\right)=\left(\begin{matrix}\ps{\ol{v_i}}{\o}\\\\x\,\vert \ol{v_i}\rangle+\ps{\ol{v_i}}{\o}\vert \ol{v_i}\rangle-\vert \o\rangle\end{matrix}\right)\,.
$$
In coordinates this gives
$$
\left(\begin{matrix}\sum_{l=1}^N{v_i^l}\,{\o^l}\\\\\left(x\, \ol{v_i^k}+\sum_{l=1}^N{v_i^l}\,{\o^l} \,\ol{v_i^k}- \o^k\right)_k\,\end{matrix}\right)\,.
$$

Let us compare it to the action of multiplication by $\sum_{k=1}^N \a_k\,X^k$ on the corresponding generic vector, in the basis $\{X^0, X^1,\ldots,X^N\}$, where $X$ is the obtuse random variable related to the $v_j$'s. We get
\begin{align*}
\left(\sum_{k=1}^N \a_k\, X^k\right)&\left(x\, X^0+\sum_{l=1}^N \o^l\, X^l\right)=\\
&=\sum_{k=1}^N\a_k\,x\, X^k+\sum_{k,l=1}^N\a_k\,\o^l\,X^k\,X^l\\
&=\sum_{k=1}^N\a_k\,x\, X^k+\sum_{k,l=1}^N\a_k\,\o^l\,\sum_{m=0}^NS^{kl}_m\,X^m\\
&=\sum_{k=1}^N\a_k\,x\, X^k+\sum_{k,l=1}^N \a_k\,\o^lS^{kl}_0\,X^0+\sum_{k,l,m=1}^N\a_k\,\o^l\,S^{kl}_m\,X^m
\end{align*}
By identification on the term in front of $X^k$ depending on $x$, we put $\a_k=\ol{v^k_i}$.
Then, with the relations \eqref{proba} and $S^{kl}_0=\EE[{X^k}\, X^l\, \ol{X^0}]$, we obtain the following term in front of $X^0$
\begin{align*}
\sum_{k,l=1}^N \ol{v^k_i}\,\o^lS^{kl}_0&= \sum_{k,l=1}^N \o^l \ol{v^k_i}\sum_{m=1}^{N+1} p_m v_m^k v_m^l\\
&=\sum_{l=1}^N \o^l \sum_{m=1}^{N+1} p_m v_m^l \sum_{k=1}^N \ol{v^k_i}v_m^k\\
&=\sum_{l=1}^N \left(\o^l \sum_{m=1}^{N+1} p_m v_m^l \left((-1)(1-\delta_{i,m}) + \normca{v_i}\,\delta_{i,m}\right)\right)\\
&=\sum_{l=1}^N \o^l p_iv_i^l + \sum_{l=1}^N \o^l p_i v_i^l \normca{v_i}\\
&=\sum_{l=1}^N \o^l p_i v_i^l(1+\normca{v_i}) = \sum_{l=1}^N \o^l v_i^l\,.
\end{align*}
Thus the terms in front of $X^0$ are equal. This gives 
\begin{align*}
\left(\sum_{k=1}^N \a_k \,X^k\right)&\left(x X^0+\sum_{l=1}^N \o^l \,X^l\right)=\\
&=\sum_{k=1}^N\ol{v^k_i}\,x\, X^k+\sum_{l=1}^N \o^l v_i^l\,X^0+\sum_{k,l,m=1}^N\a_k\,\o^l\,S^{kl}_m\,X^m\,.
\end{align*}
In the third term of the right hand side we have the expression
\begin{align*}
\sum_{l,k=1}^N \ol{v_i^k}\o^l\,S^{kl}_m&=\sum_{l,k=1}^N\sum_{\m=1}^{N+1} \ol{v^k_i}\,\o^l\,p_\m\, v^k_\m\,v^l_\m\,\ol{v^m_\m}\\
&=\sum_{l=1}^N\sum_{\m=1}^{N+1} \o^l \,p_\m \,v^l_\m\,\ol{v^m_\m}\sum_{k=1}^N\ol{v^k_i}\,v^k_\m\\
&=\sum_{l=1}^N\sum_{\m=1}^{N+1} \o^l \,p_\m \,v^l_\m\,\ol{v^m_\m}\left((-1)(1-\delta_{i,\m})+\norme{v_i}^2\, \delta_{i,m}\right)\\
&=\sum_{l=1}^N\sum_{\m=1}^{N+1}(-1)\,\o^l\,p_\m\,v^l_\m\,\ol{v^m_\m}+\sum_{l=1}^N \o^l\,p_i\,v^l_i\,\ol{v^m_i}+\sum_{l=1}^N\normca{v_i}\,\o^l\,p_i\,v^l_i\,\ol{v^m_i}\\
&=\sum_{l=1}^N -\o^l\,\d_{l,m}+\sum_{l=1}^N\o^l\,v^l_i\,\ol{v^m_i}\\
&=-\o^m+\ps{\ol{v_i}}{\o}\,\ol{v^m_i}\,.
\end{align*}
This gives the right coefficient again. We have proved the announced representation of $U$.
\end{proof}

\section{Representation of Repeated Quantum Interactions}

\subsection{Repeated Quantum Interactions}

Let us describe
precisely the physical and the mathematical setup of these models.

\smallskip
We consider a reference quantum system with state space $\rH_S$, which
we shall call the \emph{small system} (even if it is infinite dimensional!). Another system $\rH_E$, called the \emph{environment} is made up of a
chain of identical copies of a quantum system $\rK$, that is,
$$
\rH_E=\bigotimes_{n\in\NN^*} \rK
$$
where the countable tensor product is understood in a sense that we
shall make precise later. 

The dynamics in between $\rH_S$ and $\rH_E$ is obtained as follows. The
small system $\rH_S$ interacts with the first copy $\rK$ of the chain
during an interval $[0,h]$ of time. As usual the interaction is described by some Hamiltonian $H_{\rm{tot}}$ on
$\rH_S\otimes\rK$, that is, the two systems evolve together following
the unitary operator 
$$
U=e^{-ihH_{\rm{tot}}}\,.
$$
After this first interaction, the small system $\rH_S$ stops
interacting with the first copy and starts an interaction with the
second copy which was left unchanged until then. This second
interaction follows the same unitary operator $U$. And so on, the
small system $\rH_S$ interacts repeatedly with the elements of the
chain one after the other, following the same unitary evolution $U$. 
Note that with our assumptions, the unitary operator $U$ can be any unitary operator on $\rH_S\otimes\rK$. 
Let us now give a mathematical setup to this repeated quantum interaction
model.

Let $\rH_S$ be a separable Hilbert space and $\rK$ be a finite dimensional Hilbert space. We choose a fixed orthonormal
basis of $\rK$ of the form $\{X^i\,;\ i\in\rN\cup\{0\}\}$ where $\rN=\{1,\ldots,
N\}$, that is, the dimension of $\rK$ is $N+1$ (note
the particular role played by the vector $X^0$ in our notation). We
consider the Hilbert space
$$
\TF=\bigotimes_{n\in\NN^*}\rK
$$
where this countable tensor product is understood with respect to the
stabilizing sequence $(X^0)_{n\in\NN^*}$. This is to say that, denoting by $X^i_n$ the vector corresponding to $X^i$ in the $n$-th copy of $\rK$, an
orthonormal basis of $\TF$ is made of the vectors 
$$
X_\s=\bigotimes_{n\in\NN^*}X_n^{i_n}
$$
where $\s=(i_n)_{n\in\NN^*}$ runs over the set $\rP$ of all sequences in
$\rN\cup\{0\}$ with only a finite number of terms  different
from 0. 

\smallskip
Let $U$ be a fixed unitary operator on $\rH_S\otimes\rK$. We denote by
$U_n$ the natural ampliation of $U$ to $\rH_S\otimes\TF$ where $U_n$
acts as $U$ on the tensor product of $\rH_S$ and the $n$-th copy of
$\rK$ and as the identity on the other copies of $\rK$. In
our physical model, the operator $U_n$ is the unitary operator
expressing the result of the $n$-th interaction. We also define
$$
V_n=U_n\,U_{n-1}\ldots U_1\,,
$$
for all $n\in\NN^*$ and we put $V_0=I$. Clearly, the operator $V_n$ represents the transformation of the whole system after the
$n$ first interactions. 

\smallskip
Define the elementary operators $a^i_j$, $i,j\in\rN\cup\{0\}$ on $\rK$
by
$$
a^i_j\, X^k=\delta_{i,k}\,X^j\,.
$$
Note that with Dirac notation we get $a^i_j=\vert X^j\rangle\langle X^k\vert$.
We denote by $a^i_j(n)$ their natural ampliation to $\TF$ acting on
the $n$-th copy of $\rK$ only. The operator $U$ can
now be decomposed as
$$
U=\sum_{i,j\in\rN\cup\{0\}} U^i_j\otimes a^i_j
$$
for some bounded operators $U^i_j$ on $\rH_S$. The unitarity of $U$ is then equivalent to the relations
$$
\sum_{k\in\rN\cup\{0\}} (U^k_i)^*\, U^k_j=\sum_{k\in\rN\cup\{0\}}  U^k_j\,(U^k_i)^*=\delta_{i,j}\,I\,.
$$
With this representation of $U$ the operator $U_n$,
representing the $n$-th interaction, is also given by
$$
U_n=\sum_{i,j\in\rN\cup\{0\}} U^i_j\otimes a^i_j(n)\,.
$$
With these notations, the sequence ${(V_n)}$ of
unitary operators describing the $n$ first interactions can be
written as follows:
\begin{align*}
V_{n+1}&=U_{n+1}\, V_n\\
&=\sum_{i,j\in\rN\cup\{0\}} U^i_j\otimes a^i_j(n+1)V_n\,.
\end{align*}
But, inductively, the operator $V_n$ acts only on $\rH_S$ and the $n$ first sites
of the chain $\TF$, whereas the operators $a^i_j(n+1)$ act on the
$(n+1)$-th site only. Hence they commute. In the following, we shall
drop the $\otimes$ symbol, identifying operators like $a^i_j(n+1)$
with $I_{\rH_S}\otimes a^i_j(n+1)$, identifying $U^i_j$ with $U^i_j\otimes I_{\TF}$, etc. This gives finally

\begin{equation}\label{E:Vn}
\boxed{
V_{n+1}=\sum_{i,j\in\rN\cup\{0\}} U^i_j\,V_n \,a^i_j(n+1)\,.
}
\end{equation}

This equation describes completely the evolution of the whole system. One can notice its particular form, where the step between $V_n$ and the new $V_{n+1}$ is described by a basis of all the possible transformations coming from the environment, namely the $a^i_j(n+1)$, associated to corresponding effect, namely $U^i_j$, on the small system $\rH_S$.

\subsection{Classical Unitary Actions and Classical Random Walks}

Now come back to the quantum repeated interaction setup $V_{n+1}=U_nV_n$.  We consider a repetaed interaction model where the basis unitary operator $U=\sum_{i=1}^{N+1}U_i\otimes\vert\phi_i\rangle\langle\phi_i\vert$ is a classical unitary operator. Let X be the corresponding obtuse random variable given by Theorem  \ref{representation}. We consider a sequence of i.i.d random variables $(X_n)$ where $X_1\sim X$. This way the operator $U_n$ can be written as
\begin{equation}\label{E:reprUn}
U_n=A\otimes I+\sum_{j=1}^N B_j\otimes F\,\rM_{X^j_n}\, F^*,
\end{equation}
where the operator $F\,\rM_{X^j_n}\, F^*$ acts on the n-th copy of $\bigotimes_{n\in\mathbb N^*}\mathcal K$. This way the quantum repeated interactions is described by the following recursive equation

\begin{equation}\label{E:Vnh}
\boxed{
V_{n+1}=A\otimes V_n+\sum_{j=1}^N B_j\otimes V_n\otimes F\,\rM_{X^j_n}\, F^*\,.
}
\end{equation}

In terms of associated $3$-tensors and operators $a^i_j(n)$ we get 

\begin{equation}\label{E:Vn}
\boxed{
V_{n+1}=A\otimes V_n+\sum_{j=1}^N \sum_{l,k=0}^N S_{k}^{jl}B_j\otimes V_n\otimes a^j_k(n)\,.
}
\end{equation}

In the sequel, we shall be interested in the continuous time limit of this equation.

\section{Continuous-Time Limits of Classical Unitary Actions}

%
%
%

The idea behind the continuous time limit consists in considering that the time duration $h$ goes to zero. We are then interested in the limit of the operator $V_{[t/h]}$. To this end we introduce the dependance in $h$ in the expression of $V_n$ that is
\begin{equation}\label{E:Vn1}
V_{n+1}^h=A(h)\otimes V_n^h+\sum_{j=1}^N B_j(h)\otimes V_n^h\otimes F\,\rM_{(X^h_n)^j}\, F^*\,.
\end{equation}
or
\begin{equation}\label{E:Vn2}
V_{n+1}^h=A(h)\otimes V_n^h+\sum_{j=1}^N \sum_{l,k=0}^N S_{k}^{jl}(h)B_j(h)\otimes V_n^h\otimes a^j_k(n)\,.
\end{equation}
Usually in the context of quantum repeated interactions the discrete time evolution appears under the form
\begin{equation}\label{E:Vn3}
V_{n+1}^h=\sum_{j,k}  L^j_k(h)\otimes V_n^h\otimes a^j_k(n)\,.
\end{equation}
In order to obtain a non-trivial continuous time limit it has been shown in \cite{A-P} that precise asymptotic has to be imposed to the coefficients $L_k^j$. In parallel the continuous time limit of multiplication operator $\,\rM_{(X^h_n)^j}\, \,$ have been derived in \cite{ADP}. Let us recall the main result. To this end let us introduce
$$\rM_{Z_{t}^h}=\sum_{k=1}^{[t/h]}\sqrt{h}\rM_{X_k^h}$$

\begin{theorem}\label{Mij} Assume that there exists operators $M_{ij}^k$ for all $i,j=1,\ldots, N$ and all $k=0,\ldots,N$ such that: 
$$
M^{ij}_0=\lim_{h\to0} S^{ij}_0(h)
$$
and
$$
M^{ij}_k=\lim_{h\to0}\sqrt h\, S^{ij}_k(h).
$$ 
Then the operators of multiplication $\rM_{(Z^h_t)^i}$, acting of $\Phi$, converge strongly on a certain domain $\rD$ to the operators 
\begin{equation}\label{E:rY}
\rZ^i_t=\sum_{j,k=0}^N M^{ij}_k\, a^j_k(t)\,.
\end{equation}
These operators are the operators of multiplication by $Z$, the complex martingale satisfying
\begin{equation}\label{E:ESY}
[Z^i\,,Z^j]_t=M^{ij}_0\, t+\sum_{k=1}^N M^{ij}_k\, Z^k_t
\end{equation}
and
\begin{equation}\label{E:ESY2}
[\ol{Z^i}\,,Z^j]_t=\d_{ij}\, t+\sum_{k=1}^N \ol{M^{ik}_l}\, Z^k_t\,.
\end{equation}
\end{theorem}

Under the light of this theorem, we shall see that the operators $A(h)$ and $B_j(h),j=1,\ldots,N$ have to obey precise asymptotic expansions in order to obtain a continuous time limit of \eqref{E:Vn2}. More precisely we get the following final theorem which state the continuous time limit of "classical" actions.

\begin{theorem}\label{62}
Suppose that there exists a family of operators $\{\tilde A,\tilde B_j,j=1\ldots,N\}$ such that
$$\lim_{h\rightarrow0}\frac{B_j(h)}{\sqrt{h}}=\tilde B_j$$
$$\lim_{h\rightarrow0}\frac{A(h)-I}{h}=\tilde A$$
Then the process $(V_{[t/h]}^h)$ converges to the solution of
$$dU_t=\tilde AU_t\,dt+\sum_{j=1}^N\tilde B_jU_t\,dZ^j_t.$$
\end{theorem}

\begin{proof}
The proof is a direct application of \cite{A-P} and the asymptotic conditions of Theorem \ref{Mij}. Indeed, come back to the expression
\begin{equation}\label{E:Vn}
V_{n+1}^h=A(h)\otimes V_n^h+\sum_{j=1}^N \sum_{l,k=0}^N S_{k}^{jl}(h)B_j(h)\otimes V_n^h\otimes a^j_k(n)\,.
\end{equation}
which can be written as
\begin{equation}\label{E:Vn}
V_{n+1}^h=A(h)\otimes V_n^h+\sum_{j=1,k=0}^N \sum_{l=0}^N S_{k}^{jl}(h)B_j(h)\otimes V_n^h\otimes a^j_k(n)\,.
\end{equation}
which is directly in the form of \eqref{E:Vn3}. Now, it is sufficient to show that the operators
$$\sum_{l=0}^N S_{k}^{jl}(h)B_j(h)\,\,\textrm{and}\,\,A(h)$$
satisfies the conditions of \cite{A-P}. It is clear for $A(h)$ and writing 
$$S_{k}^{jl}(h)B_j(h)=\sqrt{h}S_{k}^{jl}(h)\frac{B_j(h)}{\sqrt{h}}$$ the result is immediate.
\end{proof}

\section{Examples}

In this section, we provide few enlightening examples which should illustrate Theorem \ref{62}. In particular we concentrate of examples in dimension 2 and 3 for the environment.

\subsection{Obtuse random variable on $\CC$}

Let $X$ take the values $z_1=\r_1\,e^{i\tau_1}$ with probability $p$ and $z_2=\r_2\,e^{i\tau_2}$ with probability $q$. 
The relation $\EE[X]=0$ gives $pz_1+qz_2=0$ and thus
$$
z_2=-\frac pq\, z_1
$$
so that 
$$
z_2=-\frac pq\,\r_1\, e^{i\tau_1}\,.
$$
The relation $\EE[\ab{X}^2]=1$ gives $p\ab{z_1}^2+q\ab{z_2}^2=1$, that is,
$$
p\r_1^2+q\,\frac{p^2}{q^2}\,\r_1^2=1\,.
$$
This gives
$$
\r_1=\sqrt{\frac qp}\,.
$$
(There a choice of sign made here)
As a consequence
\begin{align*}
z_1&=\sqrt{\frac qp}\, e^{i\tau}\\
z_2&=-\sqrt{\frac pq}\, e^{i\tau}\,.
\end{align*}

Now
we consider the multiplication operator by $X$, denoted $\rM_X$. For an orthonormal basis of $L^2(\O)=\CC^2$ we have two natural choices.

The basis $\{\indic, X\}$ is an orthonormal basis. In this basis we have
\begin{align*}
X\indic&=X\\
XX&=X^2=\begin{cases} 
\frac qp\, e^{2i\tau}&\mbox{with probability }p\\
\frac pq\, e^{2i\tau}&\mbox{with probability }q\,.
\end{cases}
\end{align*}
Trying to write it as $\l \indic +\m X$ gives
$$
\l=e^{2i\tau}\qq\mbox{and}\qq\m=c_p\,e^{i\tau}\,,
$$
where $c_p$ is the usual quantity
$$
c_p=\frac{q-p}{\sqrt{pq}}\,.
$$
This gives the matrix
$$
\rM_X=
\left(
\begin{matrix}
 0 & e^{2i\tau}  \\
 1 &   c_p\,e^{i\tau}  
\end{matrix}
\right)=e^{i\tau}\,
\left(
\begin{matrix}
 1 &  e^{i\tau}   \\
 e^{-i\tau} &  c_p    
\end{matrix}
\right)\,.
$$

Another possible choice is the basis $\{\indic, \ol{X}\}$ which is also an orthonormal basis. In that case we have
\begin{align*}
X\indic&=X=\ol X\, e^{2i\tau}\\
X\ol X&=\ab{X}^2=\begin{cases} 
\frac qp&\mbox{with probability }p\\
\frac pq&\mbox{with probability }q\,.
\end{cases}
\end{align*}
Trying to write it as $\l \indic +\m X$ gives
$$
\l=1\qq\mbox{and}\qq\m=c_p\,e^{-i\tau}\,.
$$
This gives the matrix
$$
\rM_X=
\left(
\begin{matrix}
 0 & 1 \\
 e^{-2i\tau} &   c_p\,e^{-i\tau}  
\end{matrix}
\right)=e^{-i\tau}\,
\left(
\begin{matrix}
 1 &  e^{i\tau}   \\
 e^{-i\tau} &  c_p    
\end{matrix}
\right)\,.
$$

Let us finally look at classical unitaries in this context.
An orthonormal basis of $\CC^2$ is always of the form
$$
\phi=
\left(
\begin{matrix}
 \sqrt p\, e^{i\tau}   \\
\sqrt q\, e^{i\n}
\end{matrix}
\right)\,,
\qq
\psi=
\left(
\begin{matrix}
 -\sqrt q\, e^{i(\tau-\n+\xi)}   \\
\sqrt p\, e^{i\xi}
\end{matrix}
\right)\,.
$$
This gives
$$
\vert \phi\rangle \langle \phi\vert=
\left(
\begin{matrix}
p  &   \sqrt{pq}\, e^{i(\tau-\n)}  \\
\sqrt{pq}\,e^{i(\n-\tau)}  & q     
\end{matrix}
\right)\,,\qq
\vert \psi\rangle \langle \psi\vert=
\left(
\begin{matrix}
q  &   -\sqrt{pq}\, e^{i(\tau-\n)}  \\
-\sqrt{pq}\,e^{i(\n-\tau)}  & p    
\end{matrix}
\right)\,.
$$
Only one angle is appearing finally, we rename it as $\tau$:
$$
\vert \phi\rangle \langle \phi\vert=
\left(
\begin{matrix}
p  &   \sqrt{pq}\, e^{i\tau}  \\
\sqrt{pq}\,e^{-i\tau}  & q     
\end{matrix}
\right)\,,\qq
\vert \psi\rangle \langle \psi\vert=
\left(
\begin{matrix}
q  &   -\sqrt{pq}\, e^{i\tau}  \\
-\sqrt{pq}\,e^{-i\tau}  & p    
\end{matrix}
\right)\,.
$$
This can be written as
\begin{align*}
\vert \phi\rangle \langle \phi\vert&= pI+
\left(
\begin{matrix}
0  &   \sqrt{pq}\, e^{i\tau}  \\
\sqrt{pq}\,e^{-i\tau}  & q-p    
\end{matrix}
\right)=pI+\sqrt{pq}\left(
\begin{matrix}
0  &   e^{i\tau}  \\
e^{-i\tau}  & c_p    
\end{matrix}
\right)\\
&=pI+\sqrt{pq}\, e^{-i\tau}\,\rM_X
\end{align*}
and
\begin{align*}
\vert \psi\rangle \langle \psi\vert&= qI+
\left(
\begin{matrix}
0  &   -\sqrt{pq}\, e^{i\tau}  \\
-\sqrt{pq}\,e^{-i\tau}  & p-q    
\end{matrix}
\right)=qI-\sqrt{pq}\left(
\begin{matrix}
0  &   e^{i\tau}  \\
e^{-i\tau}  & c_p    
\end{matrix}
\right)\\
&=qI-\sqrt{pq}\, e^{-i\tau}\,\rM_X\,.
\end{align*}

\smallskip\noindent
As \emph{classical unitary operator} is of the form
$$
U=U_1\otimes \vert\phi\rangle\langle\phi\vert+U_2\otimes \vert\psi\rangle\langle\psi\vert
$$
and hence can be written as
\begin{align*}
U&=(pU_1+qU_2)\otimes I+\sqrt{pq}e^{-i\tau}(U_1-U_2)\otimes \rM_X\\
&=A\otimes I+B\otimes \rM_X\,.
\end{align*}
In details, this gives 
\begin{align*}
U&=\begin{cases} (pU_1+qU_2)+q(U_1-U_2)&\mbox{with probability }p\\
(pU_1+qU_2)-p(U_1-U_2)&\mbox{with probability }q
\end{cases}\\
&=\begin{cases} U_1&\mbox{with probability }p\\
U_2&\mbox{with probability }q\,.
\end{cases}
\end{align*}
Now we can plug some asymptotic conditions. In dimension $2$ only diffusive or Poisson noise can appear. As we shall see the conditions of Theorem \ref{62} shall impose asymptotic conditions on the operator $U_i$
\begin{itemize}
\item Diffusive case. For a sake of simplicity assume that 
$$\phi=
\left(
\begin{matrix}
 \sqrt p\,   \\
\sqrt q\, 
\end{matrix}
\right)\,\,\textrm{and}\,\,\psi=
\left(
\begin{matrix}
 -\sqrt q\,   \\
\sqrt p\, 
\end{matrix}
\right)$$
for fixed value of $p$ and $q$ (that is independent of the parameter $h$). In this case the multiplication operator
$$\sqrt{h}\sum_{k=1}^{[t/h]}\rM_{X_k^h}$$
converges to $a_0^1(t)+a_1^0(t)$ which is simply the multiplication operator of the usual Brownian motion on $\mathbb R$. Now Theorem \ref{62} imposes that
$$\lim_{h\rightarrow0}\frac{B(h)}{\sqrt{h}}\,\,\textrm{and}\,\,\lim_{h\rightarrow0}\frac{A(h)-I}{h}$$
Here we have
$$B(h)=\sqrt{pq}(U_1-U_2)\,\,\textrm{and}\,\,A(h)=pU_1+qU_2$$
It is then clear that the operators $U_i$ must depend on $h$. In particular assume that 
$$U_i=I+\sqrt{h}O_i+hP_i+\circ(h)$$
with $pO_1+qO_2=0$ then the limit equation read
$$dU_t=(pP_1+qP_2)U_tdt+\sqrt{pq}(O_1-O_2)U_tdW_t$$
where $(W_t)$ represents a Brownian motion.
\item Poisson case. In order to obtain Poisson noise, the probability $p$ and $q$ must depend in $h$. One can choose for example $p=1-h+\circ{h}$ and $q=h+\circ(h)$. This can be obtained for $\phi$ and $\psi$ of the 
following form
$$\phi=
\left(
\begin{matrix}
 \sqrt \frac{1}{{1+h}}\,   \\
\sqrt \frac{ h}{{1+h}}\, 
\end{matrix}
\right)\,\,\textrm{and}\,\,\psi=
\left(
\begin{matrix}
 -\sqrt \frac{ h}{{1+h}}\,   \\
\sqrt \frac{1}{{1+h}}\, 
\end{matrix}
\right).$$
\end{itemize}
In this case the multiplication operator
$$\sqrt{h}\sum_{k=1}^{[t/h]}\rM_{X_k^h}$$
converges to $a_0^1(t)+a_1^0(t)+a_1^1(t)$ which is simply the multiplication operator by the usual Poisson Process of intensity $1$. Here the asymptotic for $U_i$ are slightly different. Choose for example
$$U_1=I+hV\,\,\textrm{and}\,\,U_2=W.$$
In particular this imposes $V=-iH$ for somme Hermitian operator $H$. One can see that
$$\lim_{h\rightarrow0}\frac{B(h)}{\sqrt{h}}=I-W\,\,\textrm{and}\,\,\lim_{h\rightarrow0}\frac{A(h)-I}{h}=-I+V+W$$
The limit equation reads as
$$dU_t=(-I+V+W)U_tdt+(I-W)(dN_t-dt)=-iHU_tdt+(I-W)U_t(dN_t)$$
\bigskip
\subsection{A physical example in dimension 1}

Let us workout an example based on some concrete Hamiltonian. Let $h>0$ be a fixed parameter. Consider the following Hamiltonian on $\rH\otimes\rK$
$$
H_{\rm tot}=H\otimes I+V\otimes\left(\frac{1}{\sqrt h} a^0_1+\frac 1{\sqrt h} a^1_0+\frac {1-h}h a^1_1\right)\,,
$$
where $V$ is self-adjoint on $\rH$. In other words
$$
H_{\rm tot}=H_S\otimes I+\frac 1{\sqrt h}V\otimes \left(\begin{matrix}0&1\\1&\frac {1-h}{\sqrt h}\end{matrix}\right)\,.
$$
Consider the following orthonormal basis of $\CC^2$:
$$
\vert \phi_1\rangle=\frac{1}{\sqrt{1+h}}\left(\begin{matrix}1\\-\sqrt{h}\end{matrix}\right),\qq
\vert \phi_2\rangle=\frac{\sqrt h}{\sqrt{1+h}}\left(\begin{matrix}1\\1/\sqrt{h}\end{matrix}\right)\,,
$$
so that
$$
\vert\phi_1\rangle\langle \phi_1\vert=\frac 1{1+h}\left(\begin{matrix}1&-\sqrt h\\-\sqrt h&h\end{matrix}\right)\qq\mbox{and}\qq\vert\phi_2\rangle\langle \phi_2\vert=\frac 1{1+h}\left(\begin{matrix}h&\sqrt h\\\sqrt h&1\end{matrix}\right)\,.
$$
Computing
$$
-\sqrt h\vert\phi_1\rangle\langle \phi_1\vert+\frac 1{\sqrt h}\vert\phi_2\rangle\langle \phi_2\vert
$$
gives
$$
\frac 1{1+h}\left(\begin{matrix}0&1+h\\1+h&\frac 1{\sqrt h}-h\sqrt h\end{matrix}\right)=
\left(\begin{matrix}0&1\\1&\frac {1-h}{\sqrt h}\end{matrix}\right)\,.
$$
This means that our Hamiltonian can be rewritten as
$$
H_{\rm tot}=H_S\otimes \left(\vert\phi_1\rangle\langle \phi_1\vert+\vert\phi_2\rangle\langle \phi_2\vert\right)+\frac 1{\sqrt h}V\otimes \left(-\sqrt h\vert\phi_1\rangle\langle \phi_1\vert+\frac 1{\sqrt h}\vert\phi_2\rangle\langle \phi_2\vert\right)
$$
or else
$$
H_{\rm tot}=\left(H_S-V\right)\otimes\vert\phi_1\rangle\langle \phi_1\vert+\left(H_S+\frac 1hV\right)\otimes \vert\phi_2\rangle\langle \phi_2\vert\,.
$$
When written under this form $H_{\rm tot}$ can be easily exponentiated; let us compute $U=e^{-ihH_{\rm tot}}$. We get
\begin{align*}
U&=e^{-ih(H_S-V)}\otimes\vert\phi_1\rangle\langle \phi_1\vert+e^{-ih(H_S+\frac 1hV)}\otimes \vert\phi_2\rangle\langle \phi_2\vert\\
&=e^{-ih(H_S-V)}\otimes\vert\phi_1\rangle\langle \phi_1\vert+e^{-i(hH_S+V)}\otimes \vert\phi_2\rangle\langle \phi_2\vert\,.
\end{align*}
We recover the form of a classical unitary operator on $\rH\otimes\rK$
$$
U=U_1\otimes\vert\phi_1\rangle\langle \phi_1\vert+U_2\otimes \vert\phi_2\rangle\langle \phi_2\vert\,.
$$
The matrix form for $U$ is then
$$
\frac{1}{1+h}\left(\begin{matrix}U_1+hU_2&\sqrt h(U_2-U_1)\\\sqrt h(U_2-U_1)&hU_1+U_2\end{matrix}\right)\,.
$$
Let us expand the terms of this matrix up to $o(h)$.
\begin{align*}
\frac{1}{1+h}(U_1+hU_2)&=(1-h+o(h))(I-ih(H_S-V)+he^{-iV}+o(h))\\
&=(1-h)I-ihH_S+ihV+he^{-iV}+o(h)\\
\frac{\sqrt h}{1+h}(U_2-U_1)&=(\sqrt h+o(h))(e^{-iV}-I-ihH_S+o(h))\\
&=\sqrt he^{-iV}+o(h)\\
\frac{1}{1+h}(hU_1+U_2)&=(1-h+o(h))(hI+e^{-iV}+ hK+o(h)\\
&=e^{-iV}+hK'+o(h)
\end{align*}
where $K$ and $K'$ are complicated expressions which we will not need to develop.
Altogether this gives
$$
U=\left(\begin{matrix}I+h(-iH_S+e^{-iV}-I+iV)+o(h)&&\sqrt h(e^{-iV}-I)+o(h)\\\\
\sqrt h(e^{-iV}-I)+o(h)&&I+(e^{-iV}-I)+o(1)\end{matrix}\right)\,.
$$
In the framework of \cite{A-P}, the repeated quantum interaction scheme associated to that unitary operator converges to the solution of the QSDE
$$
dV_t=(-iH_S+e^{-iV}-I+iV)V_t\, dt+(e^{-iV}-I)V_t\, (da^0_1(t)+da^1_0(t)+da^1_1(t))
$$
or else
$$
dV_t=(-iH_S+iV)V_t\, dt+(e^{-iV}-I)V_t\, dN_t
$$
where $(N_t)$ is a standard Poisson process.

\smallskip
Let us try to understand this evolution equation from the discrete-time one, directly in probabilistic terms.
The fact that
$$
U=U_1\otimes\vert\phi_1\rangle\langle \phi_1\vert+U_2\otimes \vert\phi_2\rangle\langle \phi_2\vert
$$
can be written
\begin{align*}
U&=\frac{1}{1+h}(U_1+hU_2)\otimes I+\frac{1}{1+h}\left(\begin{matrix} 0&\sqrt h(U_2-U_1)\\
\sqrt h(U_2-U_1)&(1-h)(U_2-U_1)\end{matrix}\right)\\
&=\frac{1}{1+h}(U_1+hU_2)\otimes I+\frac{\sqrt h}{1+h}(U_2-U_1)\otimes \left(\begin{matrix} 0&1\\1&\frac{1-h}{\sqrt h}\end{matrix}\right)\,.
\end{align*}
With the matrix on the right-hand side we recognize the multiplication operator by the obtuse random variable $X$ taking values $1/\sqrt h$ with probability $h/(1+h)$ and $-\sqrt h$ with probability $1/(1+h)$.

We have
$$
U=\frac{1}{1+h}(U_1+hU_2)\otimes I+\frac{\sqrt h}{1+h}(U_2-U_1)\otimes X\,.
$$
This means that $U$ takes, with probability $h/(1+h)$, the value
$$
U=\frac{1}{1+h}(U_1+hU_2)+\frac{\sqrt h}{1+h}(U_2-U_1)\frac1{\sqrt h}= U_2
$$
and, with probability $1/(1+h)$, the value
$$
U=\frac{1}{1+h}(U_1+hU_2)+\frac{\sqrt h}{1+h}(U_2-U_1)(-\sqrt h)= U_1\,.
$$
The repeated interaction scheme gives
$$
V_{n+1}=\frac{1}{1+h}(U_1+hU_2)V_n\otimes I_{n+1}+\frac{\sqrt h}{1+h}(U_2-U_1)V_n\otimes X_{n+1}
$$
or else
\begin{equation}\label{E:deltaV}
V_{n+1}-V_n=\frac{1}{1+h}(U_1+hU_2-(1+h)I)V_n\otimes I_{n+1}+\frac{\sqrt h}{1+h}(U_2-U_1)V_n\otimes X_{n+1}\,.
\end{equation}
The unitary operator $V_{n+1}$ is equal to $V_n$ multiplied on the left by $U_1$ or $U_2$ with respective probability $1/(1+h)$ and $h/(1+h)$.

This means that very often, a given value of $V_n$ will get multiplied on the left by $U_1$, more exactly it is multiplied on the left by $U_1^n$ with probability $1/(1+h)^n$. If the time steps are of length $h$ too, we get a multiplication on the right by $U_1^{t/h}$ with probability $1/(1+h)^{t/h}$. This means, in the limit $h$ tends to 0, that starting with a $V_{t_0}$, we have a random variable $\tau$ which is the first time when $X_n=1/\sqrt h$. The law of $\tau$ is a geometrical law which converge, when $h$ tends to 0 to an exponential law $\rE(1)$. Before that jumping time $\tau$, the unitary $V_{t}$ is obtained by
$$
V_t=e^{-i(t-t_0)(H_S-V)}V_{t_0}\,.
$$
At the time $\tau$, the unitary $V_t$ gets multiplied by $U_2$, that is, by $e^{-iV}$ for small $h$. This is exactly what is described by the QSDE associated to $V_t$ above.

One can also understand it directly with Equation (\ref{E:deltaV}). We have
$$
\frac{1}{1+h}(U_1+hU_2-(1+h)I)=-ihH_S+ihV+h(e^{iV}-I)+o(h)
$$
and 
$$
\frac{\sqrt h}{1+h}(U_2-U_1)=\sqrt{h}(e^{-iV}-I+o(1))
$$
so that
$$
V_{n+1}-V_n=(-iH_S+iV+(e^{iV}-I)+o(1))V_n\otimes (hI_{n+1})+(e^{-iV}-I+o(1))V_n\otimes (\sqrt{h}X_{n+1})\,.
$$
The random variable $\sqrt h X_n$ takes the values $-h$ with probability $1/(1+h)$ and 1 with probability $h/(1+h)$. We have also
$$
V_{n+1}-V_n=(-iH_S+iV)V_n\otimes (hI_{n+1})+(e^{-iV}-I+o(1))V_n\otimes (\sqrt{h}X_{n+1}+hI_{n+1})\,.
$$
The random variable $\sqrt h X_n+hI_n$ takes the values $0$ with probability $1/(1+h)$ and $1+h$ with probability $h/(1+h)$. At the limit $h$ tends to 0 we are clearly dealing with a standard Poisson process and the evolution equation is exactly our QSDE.

\subsection{An example in dimension 3}

We finish by describing two cases in the situation where $\dim\mathcal K=3$. To our best of knowledge, there exists no parametrization of orthonormal basis of $\mathbb C^3$ as the one we have described for $\mathbb C^2$. This way we are not able to describe the whole situation and we restrict ourself to two examples which we think are sufficiently illustrating. In the two below examples we describe directly the obtuse system $v_i\,,\ i=1,2,3$ involved in Theorem \ref{representation}.

\begin{itemize}
\item First example. Consider the following vectors 
$$v_1=\left(\begin{array}{cc}1\\0\end{array}\right),\quad v_2=\left(\begin{array}{cc}-1\\1\end{array}\right),\quad v_3=\left(\begin{array}{cc}-1\\-2\end{array}\right).$$
On can check that the involved probabilities are
$$p_1=p=\frac 12,q_1=q=\frac 13,r_1=r=\frac16.$$
Here we assume that the unitary operators follows the following assumptions
$$U_i=I+\sqrt{h}O_i+hP_i.$$
In order to obtain an effective limit we impose that
$$\frac{O_1}{2}+\frac{O_2}{3}+\frac{O_3}{6}=0.$$
In this case we have the following expression for $A(h),B_1(h),B_2(h)$:
\begin{eqnarray*}A&=&pU_1+qU_2+rU_3=I+h\left(\frac{P_1}{2}+\frac{P_2}{3}+\frac{P_3}{6}\right)\\
B_1&=&\sqrt{h}\left(\frac{O_1}{2}-\frac{O_2}{3}-\frac{O_3}{6}\right)\\
B_2&=&\sqrt{h}\left(\frac{O_2}{3}-\frac{O_3}{3}\right)
\end{eqnarray*}
In that case, the operators of multiplication $\rM_{(Z^h_t)^i},i=1,\ldots,2$ converge to operators of multiplication $\rM_{(Z_t)^i},i=1,\ldots,2$ where
$$
\begin{cases}
Z^1_t&= W_1(t)\\\\
Z^2_t&=W_2(t)\,,
\end{cases}
$$
with $(W_1,W_2)$ two independent Brownian motions. It turns out that the continuous limit evolution satisfies
$$dU_t=\left(\frac{P_1}{2}+\frac{P_2}{3}+\frac{P_3}{6}\right)dt+\left(\frac{O_1}{2}-\frac{O_2}{3}-\frac{O_3}{6}\right)dW_1(t)+\left(\frac{O_2}{3}-\frac{O_3}{3}\right)dW_2(t).$$
\item Second example. In this example we aim at mixing Poisson process and Brownian motion. This can be done by considering
$$
v_1=\frac 1{\sqrt 2}\left(\begin{matrix} i\\\ecarte 1\end{matrix}\right)\ ,\qq 
v_2=\frac 1{\sqrt{2h}}\left(\begin{matrix}{1-i\sqrt{h}}\\\ecarte{i-\sqrt{h}}\end{matrix}\right)\ ,\qq 
v_3=\frac 1{\sqrt 2}\left(\begin{matrix}{-2\sqrt{h}-i}\\\ecarte{-1-2i\sqrt h}\end{matrix}\right),
$$
with associated probabilities.
$$p_1=1/2,p_2=h/(1+2h),p_3=1/(2+4h).$$
For the unitary operators we assume that
$$U_1=I+\sqrt{h}O_1+hP_1,U_2=U_2,U_3=I+\sqrt{h}O_3+hP_3$$
such that
$$O_1+O_3=0.$$
Then we get the following expression
\begin{eqnarray*}A&=&I+h\left(\frac{P_1}{2}+U_2+\frac{P_3}{2}-I\right),\\
B_1&=&\frac{-i}{2\sqrt{2}}U_1+\frac{h(1+i\sqrt h)}{\sqrt 2(1+2h)\sqrt h}U_2+\frac{-2\sqrt h+i}{(2+4h)\sqrt 2}U_3\\
&=&\sqrt h\left(\frac{-i}{2\sqrt{2}}O_1+\frac{U_2-I}{\sqrt 2}+\frac{i}{2\sqrt 2}O_3\right),\\
B_2&=&\frac{1}{2\sqrt{2}}U_1+\frac{h(-i-\sqrt h)}{\sqrt 2(1+2h)\sqrt h}U_2+\frac{-1+2i\sqrt h}{(2+4h)\sqrt 2}U_3\\
&=&\sqrt h\left(\frac{1}{2\sqrt{2}}O_1-i\frac{U_2-I}{\sqrt 2}+\frac{-1}{2\sqrt 2}O_3\right).
\end{eqnarray*}
Moreover, in that case the operators of multiplication $\rM_{(Z^h_t)^i},i=1,\ldots,2$ converge to operators of multiplication $\rM_{(Z_t)^i},i=1,\ldots,2$ where
$$
\begin{cases}
Z^1_t&=\frac{1}{\sqrt{2}}\, (N_t-t)+i W_t\\\\
Z^2_t&=\frac{i}{\sqrt{2}}\, (N_t-t)+\, W_t\,,
\end{cases}
$$
with $N$ a Poisson process of intensity $1$ independent of the Brownian motion $W$. The limit equations can be then computed and we get:
\begin{eqnarray}dU_t&=&\left(\frac{P_1}{2}+U_2-I+\frac{P_3}{2}\right)dt+\\&&+\left(\frac{-i}{2\sqrt{2}}O_1+\frac{U_2-I}{\sqrt 2}+\frac{i}{2\sqrt 2}O_3\right)dZ_1(t)+\left(\frac{1}{2\sqrt{2}}O_1-i\frac{U_2-I}{\sqrt 2}+\frac{-1}{2\sqrt 2}O_3\right)
dZ_2(t)\nonumber\\
&=&\left(\frac{P_1}{2}+\frac{P_3}{2}\right)dt+(I-U_2)dN_t+\frac{1}{\sqrt 2}(O_1-O_3)dW_t\,.
\end{eqnarray}
 \end{itemize}

\centerline{\timesept St\'ephane ATTAL}
\vskip -1mm
\centerline{\timesept Universit\'e de Lyon}
\vskip -1mm
\centerline{\timesept Universit\'e de Lyon 1, C.N.R.S.}
\vskip -1mm
\centerline{\timesept Institut Camille Jordan}
\vskip -1mm
\centerline{\timesept 21 av Claude Bernard}
\vskip -1mm
\centerline{\timesept 69622 Villeubanne cedex, France}
\vskip 2mm
\centerline{\timesept Julien DESCHAMPS}
\vskip -1mm
\centerline{\timesept Universit\`a degli Studi di Genova}
\vskip -1mm
\centerline{\timesept Dipartimento di Matematica}
\vskip -1mm
\centerline{\timesept Via Dodecaneso 35}
\vskip -1mm
\centerline{\timesept 16146 Genova, Italy}
\vskip 2mm
\centerline{\timesept Cl\'ement PELLEGRINI}
\vskip -1mm
\centerline{\timesept Institut de Math\'ematiques de Toulouse }
\vskip -1mm
\centerline{\timesept Equipe de Statistique et de Probabilit\'e}
\vskip -1mm
\centerline{\timesept Universit\'e Paul Sabatier (Toulouse III)}
\vskip -1mm
\centerline{\timesept 31062 Toulouse Cedex 9, France}


\begin{thebibliography}{99}

\bibitem{ADP} S. Attal, J. Deschamps, C. Pellegrini, "Complex Obtuse Random Walks and Their Continuous-Time Limits", \emph{Probability Theory and Related Fields}, to appear.

\bibitem{Bip}  S. Attal, J. Deschamps, C. Pellegrini, "Entanglement of Bipartite Quantum Systems driven by Repeated Interactions", \emph{Journal of Statistical Physics},  154 (2014), no. 3, 819-837.

\bibitem{Att} S. Attal, "Approximating the Fock space with the toy Fock space", 
\emph{S\'eminaire de Probabilit\'es} XXXVI, Springer L.N.M. 1801 (2003) , p. 477-497.

\bibitem{A-D} S. Attal, A. Dhahri, ``Repeated Quantum Interactions and Unitary Random Walks", \emph{Journal of Theoretical Probability}, 23, p. 345-361, 2010. 

\bibitem{A-E} S. Attal, M. Emery, ``Equations de structure pour des martingales vectorielles", \emph{S\'eminaire de Probabilit\'es}, XXVIII, p. 256--278, Lecture Notes in Math., 1583, Springer, Berlin, 1994. 

\bibitem{A-P} S. Attal, Y. Pautrat, ``From repeated to continuous quantum interactions",
\emph{Annales Henri Poincar\'e. A Journal of Theoretical and Mathematical Physics}, 7, 2006, p. 59--104.
		
\bibitem{A-P2} S. Attal, Y. Pautrat, ``From (n+1)-level atom chains to $n$-dimensional noises", \emph{Ann. Inst. H. Poincar\'e Probab. Statist.} 41 (2005), no. 3, p. 391--407. 

\bibitem{B-P} L. Bruneau, C.-A. Pillet, ``Thermal relaxation of a QED cavity", 
\emph{J. Stat. Phys.} 134 (2009), no. 5-6, p. 1071--1095. 

\bibitem{BdP} L. Bruneau, S. De Bi\`evre, C.-A. Pillet, ``Scattering induced current in a tight-binding band" , \emph{J. Math. Phys.} 52 (2011), no. 2, 022109, 19 pp. 

\bibitem{DNP}J. Deschamps, I. Nechita and C. Pellegrini, ``On some classes of bipartite unitary operators"
\emph{J. of Physics A: Mathematical and Theoretical}, vol 49 , Number 33 (2016).

\bibitem{Har} S. Haroche, S. Gleyzes, S. Kuhr, C. Guerlin, J. Bernu, S. Del\'eglise, U. Busk-Hoff, M. Brune and J-M. Raimond, ``Quantum jumps of light recording the birth and death of a photon in a cavity", \emph{Nature} 446, 297 (2007) 

\bibitem{Har2} S. Haroche, C. Sayrin, I. Dotsenko, XX. Zhou, B. Peaudecerf, T. Rybarczyk, S. Gleyzes, P. Rouchon, M. Mirrahimi, H. Amini, M.Brune and J-M. Raimond, ``Real-time quantum feedback prepares and stabilizes photon number states", \emph{Nature}, 477, 73 (2011) 

\bibitem{H-H} Y. P. Hong, R. A. Horn, ``On Simultaneous Reduction of 
Families of Matrices to Triangular or Diagonal Form by Unitary Congruence", \emph{Linear and Multilinear Algebra}, 17 (1985), p.271-288. 

\bibitem{H-P} R.L. Hudson, K.R. Parthsarathy, "Quantum Stochastic Evolutions and ...", \emph{Communications in Mathematical Physics} ....

\bibitem{Pel} C. Pellegrini, ``Existence, uniqueness and approximation of a stochastic Schr\"odinger equation: the diffusive case", \emph{Ann. Probab.} 36 (2008), no. 6, p. 2332--2353. 

\bibitem{Pel2} C. Pellegrini, ``Existence, Uniqueness and Approximation of the jump-type Stochastic Schr\"odinger Equation for two-level systems",  \emph{Stochastic Process and their Applications}, 2010 vol 120 No 9, pp. 1722-1747.

\bibitem{Pel3} C. Pellegrini, ``Markov Chain Approximations of Jump-Diffusion Stochastic Master Equations", \emph{Annales de l'institut Henri Poincar\'e: Probabilit\'es et Statistiques}, 2010, vol 46, pp. 924-948.



\end{thebibliography}
\end{document}